\documentclass[journal]{IEEEtran}

\usepackage{graphicx,xcolor}
\usepackage{mathtools,amssymb,mathrsfs,amsfonts,dsfont}

\usepackage[standard,amsmath,thmmarks]{ntheorem}
\newtheorem{problem}{Problem}

\usepackage[sort,compress]{cite}
\usepackage{graphicx}
\usepackage{subcaption}
\usepackage{footmisc}

\newcommand\EXP{\mathds{E}}
\newcommand\reals{\mathds R}
\newcommand\integers{\mathds Z}

\begin{document}

\title{Counterexamples on the monotonicity of delay optimal strategies for
       energy harvesting transmitters}

\author{Borna Sayedana and Aditya Mahajan%
\thanks{The authors are with
	the Department of Electrical and Computer Engineering, McGill
	University, Montreal, QC, Canada. Email: {borna.sayedana@mail.mcgill.ca, aditya.mahajan@}mcgill.ca. This work was supported in part by Natural Sciences and Engineering
Research Council of Canada (NSERC) Discovery Grant RGPIN-2016-05165.}}

\maketitle

\begin{abstract}
  We consider cross-layer design of delay optimal transmission strategies
  for energy harvesting transmitters where the data and energy arrival
  processes are stochastic. Using Markov
  decision theory, we show that the value function is weakly increasing in
  the queue state and weakly decreasing in the battery state. It is natural to
  expect that the delay optimal policy should be weakly increasing in the queue and
  battery states. We show via counterexamples that this is not the case.
  In fact, we show that for some sample scenarios the delay optimal policy may perform 5--13\% better than the best \emph{monotone} policy.
\end{abstract}


\section{Introduction}

Latency is an important consideration in many Internet of
Things (IoT) applications which provide real-time and/or critical services.
Often IoT devices are battery powered and harvest
energy from the environment. In such situations, intelligent transmission
strategies are needed to mitigate the unreliability of
available energy and provide low-latency services. 

In this paper, we investigate the cross-layer design of delay optimal
transmission strategies for energy harvesting transmitters when both
the data arrival and the energy arrival processes are stochastic. Our
motivation is to characterize qualitative properties of optimal transmission
policies for such model. For example, in queuing theory, it is often
possible to establish that the optimal policy is monotone increasing in the
queue length \cite{stidham1989monotonic,gallisch1979monotone}. Such a
property, in addition to being intuitively satisfying, simplifies the search
and implementation of the optimal strategies. Such monotonicity properties are
also known to hold for cross-layer design of communication systems when a
energy is always available at the
transmitter~\cite{berry2000power}. So it is natural to ask if such qualitative
properties hold for energy harvesting transmitters. 

Partial answers to this question for throughput optimal policies for energy
harvesting transmitters are provided in \cite{zafer2008optimal,
ahmed2016optimal, sinha2012optimal, mao2014joint,kashef2012optimal,
shaviv2018online}. Under the assumptions of backlogged traffic or 
deterministic data arrival process or 
deterministic energy arrival process, these papers show that the optimal
policy is weakly increasing in the queue state and/or weakly increasing in the
battery state. There are other papers that investigate the structure of delay
or throughput optimal policies under the assumption of a deterministic energy
arrival process~\cite{yang2012optimal,tutuncuoglu2012optimum,ozel2011transmission}.

There are some papers which investigate the problem of delay optimization for
energy harvesting transmitters~\cite{ozel2011transmission, Fawaz2018,
Sharma2018}, but they don't characterize the structure of delay-optimal
policies rather provide numerical solutions or propose low-complexity
heuristic policies or only establish structural properties of value functions.

We show that the delay optimal policy for energy harvesting communication
systems is not necessarily monotone in battery or queue state. This is in
contrast to the monotonicity of delay optimal policies when energy is always
available~\cite{berry2000power} or throughput optimal policies for energy
harvesting models~\cite{zafer2008optimal, ahmed2016optimal, sinha2012optimal,
mao2014joint,kashef2012optimal, shaviv2018online}. We present counterexamples
to show that the delay optimal policy need not be weakly increasing in queue
or battery state. Furthermore, for some sample scenarios, the
performance of the optimal policy is about 5--13\% better than that
of the best monotone policy. These counterexamples continue to
hold for i.i.d.\@ fading channels as well.

\subsubsection*{Notation}

Uppercase letters (e.g., $E$, $N$, etc.) represent random variables; the
corresponding lowercase letters (e.g., $e$, $n$, etc.) represent their
realizations. Cursive letters (e.g., $\mathcal L$, $\mathcal B$, etc.)
represent sets. The sets of real, positive integers, and non-negative integers
are denoted by $\reals$, $\integers_{> 0}$, and $\integers_{\ge 0}$
respectively. The notation $[a]_{L}$ is a short hand for $\min \{a,L\}$.

\section{Model And Problem Formulation} \label{sec:model}

\begin{figure}[!t]
  \centering
  \includegraphics[width=0.85\linewidth]{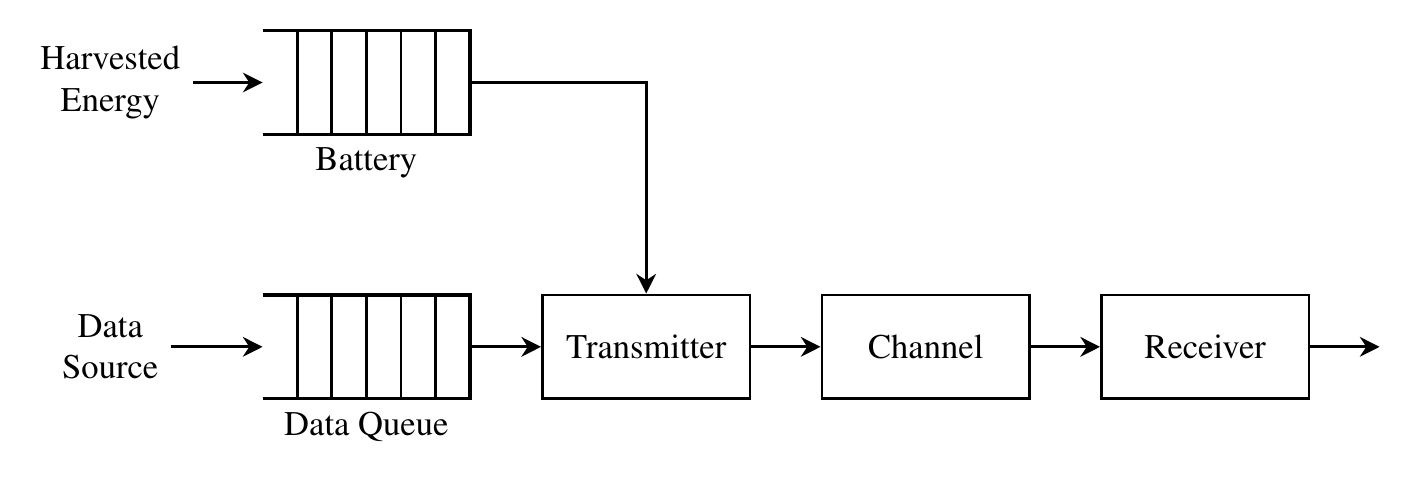}
  \caption{Model of a transmitter with energy-harvester}
  \label{fig:model}
\end{figure}

Consider a discrete-time communication system shown in Fig~\ref{fig:model}. A source
generates bursty data packets that have to be transmitted to a receiver by an
energy-harvesting transmitter. The transmitter has finite buffer where the
data packets are queued and a finite capacity battery where the harvested
energy is stored. 

At the beginning of a slot, the transmitter picks some data packets from
the queue, encodes them, and transmits the encoded symbol. Transmitting a
symbol requires energy that depends on the number of encoded packets in the
symbol. At the end of the slot, the system incurs a delay penalty that
depends on the number of packets remaining in the queue.

Time slots are indexed by $k \in \integers_{\geq 0}$. The length of the
buffer is denoted by $L$ and the size of the battery by $B$; $\mathcal{L}$
and $\mathcal B$ denote the sets $\{0,1,\ldots,L\}$ and $\{0,1,\ldots,
B\}$, respectively. Other variables are as follows:
\begin{itemize}
  \item $N_{k} \in \mathcal{L}$:  the number of packets in the
    queue at the beginning of slot~$k$.
  \item $A_{k} \in \mathcal{L}$: the number of packets that arrive
    during slot~$k$.
  \item $S_{k} \in \mathcal{B}$: the energy stored in the battery at
    the beginning of slot~$k$.
  \item $E_{k} \in \mathcal{B}$:  the energy that is harvested during
    slot~$k$.
  \item $U_{k}$: the number of packets transmitted during slot
    $k$. The feasible choices of $U_{k}$ are denoted by $\mathcal{U}(N_k,
    S_k)$
    where
    \[
      \mathcal{U}(n,s) \coloneqq \{ u \in \mathcal{L} : 
      u \le n \text{ and } p(u) \le s \},
    \]
    where $p(u)$ denotes the amount of power needed to transmit
    $u$ packets.\footnote{In our
    examples, we model the channel as a band-limited AWGN channel with
    bandwidth $W$ and noise level $N_0$. The capacity of such a channel
    when transmitting at power level $P$ is $W \log_2( 1 + P/(N_0W))$.
    Therefore, for such channels we assume $p(u) = \lfloor N_0 W (2^{u/W} -
  1) \rfloor$.} 
    We assume that $p:\mathcal{L} \rightarrow \mathds{R}_{\geq 0}$ is a
    strictly convex and increasing function with $p(0) = 0$. 
\end{itemize}

The dynamics of the data queue and the battery are
\[
  N_{k+1} = [N_{k} - U_{k} + A_{k}]_{L} 
  \quad\text{and}\quad
  S_{k+1} = [S_{k} - p(U_{k})+E_{k}]_{B}.
\]

Packets that are not transmitted during slot $k$ incur a delay penalty
$d(N_{k} - U_{k})$, where $d:\mathcal{L} \rightarrow \mathds{R}_{\geq 0}$ is
a convex and increasing function with $d(0) = 0$.

The data arrival process $\{A_{k}\}_{k\geq 0}$ is i.i.d.\@
with pmf (probability mass function) $P_A$.
The energy arrival process $\{E_{k}\}_{k \geq 0}$ is i.i.d.\@ with  pmf $P_E$
and is also independent of $\{A_{k}\}_{k \geq 0 }$.

The number  $U_{k}$ of packets to transmit are chosen according to a
scheduling policy $f \coloneqq \{f_{k}\}_{k \geq 0}$, where 
\[
  U_{k} = f_{k}(N_{k},S_{k}), \quad U_k \in \mathcal{U}(N_k, S_k).
\]

The performance of a scheduling policy $f$ is given by
\begin{equation} \label{eq:cost}
  J(f) \coloneqq \EXP^{f} \Big[ \sum_{k = 0}^{\infty} \beta^{k} d(N_{k}-U_{k})
  \Bigm| N_{0} = 0,S_{0} =0 \Big],
\end{equation}
where $\beta \in (0,1)$ denotes the discount factor and the expectation is taken with respect to the joint measure on the system variables induced by the choice of $f$.

We are interested in the following optimization problem.
\begin{problem}\label{prob:main}
  Given the buffer length $L$, battery size $B$, power cost $p(\cdot)$, delay
  cost $d(\cdot)$, pmf $P_A$ of the arrival process, pmf $P_E$ of the energy arrival
  process, and the discount factor $\beta$, choose a feasible scheduling
  policy $f$ to minimize the performance $J(f)$ given by~\eqref{eq:cost}.
\end{problem}

\section{Dynamic Programming Decomposition}\label{sec:DP}

The system described above can be modeled as an infinite horizon time
homogeneous Markov decision process (MDP)~\cite{puterman2014markov}. Since
the state and action spaces are finite, standard results from Markov decision
theory imply that there exists an optimal policy which is time homogeneous and
is given by the solution of a dynamic program. To succinctly write the dynamic
program, we define the following Bellman operator:
Define the operator $\mathscr B : [\mathcal{L} \times \mathcal{B} \to
\reals] \to [\mathcal{L} \times \mathcal{B} \to \reals]$ that maps
any $V: \mathcal{L} \times \mathcal{B} \rightarrow \mathds{R}$ to
\begin{multline}\label{DyP}
  \big[\mathscr{B} V\big](n,s) = 
  \min_{u \in \mathcal{U}(n,s)}\Big\{ d(n-u) \\+ 
    \beta\EXP\big[ V([n - u + A]_{L}, [s - p(u) + E]_{B}) \big]
  \Big\},
\end{multline}
where $A$ and $E$ are independent random variables with pmfs $P_A$ and $P_E$. 
Then, an optimal policy for the infinite horizon MDP is given as
follows~\cite{puterman2014markov}. 

\begin{theorem} 
  Let $V^*: \mathcal {L}\times \mathcal {B} \to \reals$ denote the unique
  fixed point of the following equation:
  \begin{equation}\label{eq:bellman}
    V(n,s) = [\mathscr{B} V](n,s),
    \quad \forall (n,s) \in \mathcal L \times \mathcal B.
  \end{equation}
  Furthermore, let $f^*$ be such that $f^*(n,s)$ attains the minimum in the
  right hand side of~\eqref{eq:bellman}.
  Then, the time homogeneous policy $f^{*,\infty} = (f^*, f^*, \dots)$ is optimal
  for Problem~\ref{prob:main}.
\end{theorem}

The dynamic program described in \eqref{eq:bellman} can be solved using
standard algorithms such as value
iteration,  policy iteration, or linear programming
algorithms~\cite{puterman2014markov}.

\subsection{Properties of the value function}

Let $\mathcal {M}$ denote the family of the functions $V:{\cal{L}} \times
{\cal{B}} \rightarrow \mathds{R}$ such that for any $s \in {\mathcal{B}}$,
$V(n,s)$ is weakly increasing in $n$ and for any $n \in \mathcal{L}$, $V(n,s)$
is weakly decreasing in $s$.
Furthermore, let $\mathcal F_{s}$ denote the family of functions $f \colon
\mathcal L \times \mathcal B \to \mathcal U$ such that  for any $n \in \mathcal{L}$, $f(n,s)$ is weakly increasing in $s$. Similarly, let $\mathcal F_{n}$ be family of functions $f \colon
\mathcal L \times \mathcal B \to \mathcal U$, such that for any $s \in \mathcal B$, $f(n,s)$ is weakly increasing in $n$.

\begin{proposition} \label{thm:monotone}
  The optimal value function $V^* \in \mathcal{M}$.
\end{proposition}
The proof is presented in the Appendix.
Proposition~\ref{thm:monotone} says the optimal cost weakly increases with the
queue state and weakly decreases with the battery state. Thus, it's better to
have less packets in the queue and it is better to have more energy in the
battery. Such a result is intuitively appealing.

One might argue that it should be the case that the optimal policy should be weakly increasing in state of the queue, and weakly increasing in the available energy in the battery. In particular, if it is optimal to transmit $u$
packets when the queue state is $n$, then (for the same battery state) the
optimal number of packets to transmit at any queue state larger than $n$
should be at least $u$. Similarly, if it is optimal to transmit $u$ packets
when the battery state is $s$, then (for the same queue state) the optimal
number of packets to transmit at any battery state larger than~$s$ should be
at least $u$. In the next section, we present counterexamples that show 
both of these properties do not hold. The code for all the results is available at \cite{code}.

\section{Counterexamples on the monotonicity of optimal policies}
\label{sec:counterexample}

\subsection{On the monotonicity in queue state} \label{sec:queue}

\begin{figure}[!tb]
  \hbox to \linewidth{
    \hss
    \begin{subfigure}[t]{0.35\linewidth}
      \includegraphics[page=1,scale=0.95]{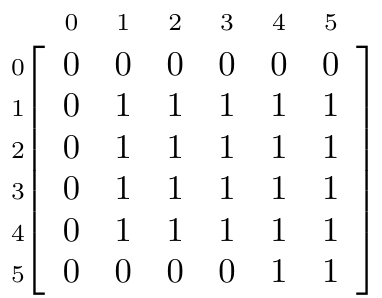}
      \caption{The optimal policy}\label{fig:ex1a}
    \end{subfigure}%
    \hskip 3em 
    \begin{subfigure}[t]{0.35\linewidth}			
      \includegraphics[page=2,scale=0.95]{policies.pdf}
      \caption{The best monotone \rlap{policy}}\label{fig:ex1b}				
    \end{subfigure}%
  \hss}
  \caption{The optimal and the best monotone policies for the example of Sec.~\ref{sec:queue}.}
  \label{fig:ex1}
\end{figure}

Consider the communication system with a band-limited AWGN channel where
$\mathcal{L} = 5 $, $\mathcal{B} = 5 $, $\beta = 0.99$, $N_0 = 2.0$, $W = 1.75$
(thus, $p(u) = \lfloor 3.5 \cdot (2^{(u/1.75)}-1) \rfloor$), $d(q) = q$, data
arrival distribution $P_A = \text{Geom}(0.9)$, and energy arrival distribution
$P_E = \text{Geom}(0.89)$, where both the pmfs are truncated and normalized with domain equal to $5$.

The optimal policy for this system (obtained by policy
iteration~\cite{puterman2014markov}) is shown in Fig.~\ref{fig:ex1a},
where the rows correspond to the current queue length and the columns
correspond to the current energy level. Note that the policy is not weakly
increasing in queue state (i.e, $f^{*} \notin \mathcal{F}_{n}$). For instance, $f(5,3) < f(4,3)$.

Given that the optimal policy is not monotone, one might wonder how much do we
lose if we use a monotone policy instead of the optimal policy. To
characterize this, we define the best queue-monotone policy as:
\[
  f^{\circ}_{n} = \arg\min_{f \in \mathcal{F}_{n}} 
  \left\{ \max_{(n,s) \in \mathcal L \times \mathcal B} 
  \bigl|   V(n,s) - V^*(n,s) \bigr|   \right\}
\]
and let $V^{\circ}_{n}$ denote the corresponding value function. 

The best monotone policy cannot be obtained using dynamic programming and one
has to resort to a brute force search over all monotone policies. For the
model described above, there are $86400$ monotone
policies.\footnote{\label{fnt}Due to the power constraint $\mathcal U(n,s)$, it
  is not possible to count the number of monotone functions using
combinatorics. The number above is obtained by explicit enumeration.}
The best monotone policy obtained by searching over these is shown is
Fig.~\ref{fig:ex1b}.
The worst case difference between the two value functions is given
by
\[
  \alpha_{n} = \max_{(n,s) \in \mathcal L \times \mathcal B}
  \left\{
    \frac{ V^\circ(n,s) - V^*{n}(n,s)}
    { V^*(n,s) }
  \right\}
  = 0.1186.
\]

Thus, for this counterexample, the best queue-monotone policy performs $11.86\%$
worse than the optimal policy. 

We also compare the performance of the optimal policy with the greedy
policy, which is a heuristic policy that transmits the maximum number of
packets in each state. The greedy policy is monotone so we
expect $\alpha_n^{\text{greedy}} \ge \alpha_n$. In this particular example, we find that $\alpha_n^{\text{greedy}} = 0.8609$. Thus, the greedy policy performs $86.09\%$ worse than the optimal policy.

\subsection{On the monotonicity in battery state} \label{sec:battery}
Consider the communication system described in Sec.~\ref{sec:queue} but with
the data arrival distribution $P_A=[0.33,0.67,0,0,0]$ and the energy arrival distribution $P_E = [0.05,0.90,0.05,0,0]$.

The optimal policy (obtained using policy
iteration~\cite{puterman2014markov}) is shown in Fig.~\ref{fig:ex2a}.
Note that the policy is not weakly increasing in the battery state (i.e $f^{*} \notin \mathcal{F}_{s}$). In particular, we have that $f^{*}(5,2) > f^{*}(5,3)$.

\begin{figure}[!tb]
  \hbox to \linewidth{
    \hss
    \begin{subfigure}[t]{0.35\linewidth}
      \includegraphics[page=3,scale=0.95]{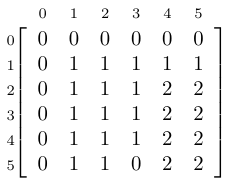}
      \caption{The optimal policy}\label{fig:ex2a}
    \end{subfigure}%
    \hskip 3em 
    \begin{subfigure}[t]{0.35\linewidth}			
      \includegraphics[page=4,scale=0.95]{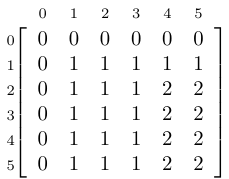}
      \caption{The best monotone \rlap{policy}}\label{fig:ex2b}				
    \end{subfigure}
  \hss}
  \caption{The optimal and the best monotone policies for the example of Sec.~\ref{sec:queue}.}
  \label{fig:ex2}
\end{figure}

Given that optimal policy is not monotone, the previous question arises again that how much do we lose if we use a monotone policy instead of the optimal policy. To characterize this, we define the best battery-monotone policy as:
\begin{align*}
  f^{\circ}_{s} = \arg\min_{f \in \mathcal{F}_{s}} 
  \left\{ \max_{(n,s) \in \mathcal L \times \mathcal B}    
    \bigl|   V(n,s)-V^*(n,s) \bigr|   \right\}
\end{align*}
and let $V^{\circ}_{s}$ denote the corresponding value function.

As before, we find the best monotone policy by a a brute force search over
all $303750$ monotone battery-policies.\footref{fnt} The resultant policy is shown in Fig.~\ref{fig:ex2b}.

The worst case difference between the two value functions is given
by
\[
  \alpha_{s} = \max_{(n,s) \in \mathcal L \times \mathcal B}
  \left\{
    \frac{ V^\circ(n,s) - V^*{s}(n,s)}
    { V^*(n,s)}
  \right\}
  = 0.0560.
\]

Thus, for this counterexample, the best battery-monotone policy performs $5.60\%$
worse than the optimal policy.  Note that for this example, the best monotone policy is a greedy policy, hence the performance of the greedy policy is  same as that of the best monotone policy.

\section{Counterexamples for fading channels} \label{sec:fading}

\begin{figure*}[!tb]
  \hbox to \linewidth{
    \hss
    \begin{subfigure}[t]{0.20\linewidth}
      \includegraphics[page=5,scale=0.95]{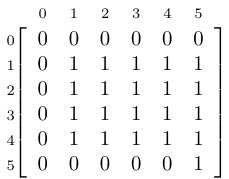}
      \caption{$f^*(\cdot, \cdot, h=1)$}\label{fig:ex3a}
    \end{subfigure}%
      \hfill
    \begin{subfigure}[t]{0.20\linewidth}
      \includegraphics[page=6,scale=0.95]{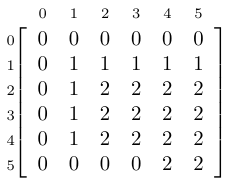}
      \caption{$f^*(\cdot, \cdot, h=2)$}\label{fig:ex3b}
    \end{subfigure}%
      \hfill
    \begin{subfigure}[t]{0.20\linewidth}
      \includegraphics[page=7,scale=0.95]{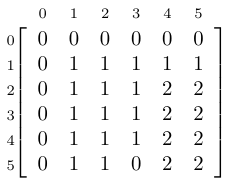}
      \caption{$f^*(\cdot, \cdot, h=3)$}\label{fig:ex4a}
    \end{subfigure}%
    \hfill
    \begin{subfigure}[t]{0.20\linewidth}
      \includegraphics[page=8,scale=0.95]{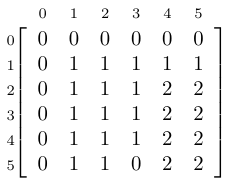}
      \caption{$f^*(\cdot, \cdot, h=1)$}\label{fig:ex4b}
    \end{subfigure}%
  \hss}
  \caption{The optimal policy for the examples of Sec.~\ref{sec:fading-queue}
    shown in subfigures (a)--(b) and Sec.~\ref{sec:fading-battery} shown in
  subfigures (c)--(d).}
\end{figure*}

\subsection{Channel model with i.i.d.\@ fading}

Consider the model in Sec.~\ref{sec:model} where the channel has i.i.d.\@
fading. In particular, let $H_k \in \mathcal H$ denote the channel state at
time~$k$ and $g(H_k)$, where $g : \mathcal H \to \reals_{> 0}$, denote the 
attenuation at state~$H_k$. Thus, the power needed to transmit $u$ packets
when the channel is in state $h$ is given by $p(u)/g(h)$. We assume
that $\{H_k\}_{k \ge 0}$ is an i.i.d.\@ process with pmf $P_H$ that is
independent of the data and energy arrival processes $\{A_k\}_{k \ge 0}$ and
$\{E_k\}_{k \ge 0}$. 

\subsection{On the monotonicity in queue state} \label{sec:fading-queue}

Consider the model in Sec.~\ref{sec:queue} with $N_{0} = 1$, $W = 1.75$, and  an i.i.d.\@ fading channel
where $\mathcal H = \{1,2\}$, $g(\cdot) = \{0.7, 0.8\}$ and $P_H =
[0.4, 0.6]$. The optimal policy for this model (obtained using policy
iteration) is shown in Fig.~\ref{fig:ex3a}--\ref{fig:ex3b}. Note that for all $h$, the optimal policy in not monotone in the queue length. 

In this case, there are  $(4320) \times (1296) \times (362) \approx 10^{8}$ 
monotone policies. Therefore, a brute force search to find the best monotone
policy is not possible. We choose a heuristic monotone policy $f^\circ_n$
which differs from $f^*$ only at the following points:
$f^\circ_n(5,s,1)=1$, for $s \in \{1,2,3,4\}$, $f^\circ_n(5,1,2) = 1$, and $f^\circ_n(5,s,2) = 2$, for $s \in \{2,3\}$.
 The policy $f^\circ_n$ may
be thought of as the queue-monotone policy that is closest to~$f^*$. Let
$V^\circ_n$ denote the corresponding value function. The worst case difference
between the two value functions is given by
\[
  \alpha_{n} = \max_{(n,s,h) \in \mathcal L \times \mathcal B \times \mathcal H}
    \frac{\bigl| V^*(n,s,h) - V^\circ_{n}(n,s,h)\bigr|}
    {\bigl| V^*(n,s,h) \bigr|}
  = 0.1344.
\]
Thus, the heuristically chosen queue-monotone policy performs $13.44\%$ worse
than the optimal policy. We also compare the optimal policy with the greedy policy and find that $\alpha_n^{\text{greedy}} = 0.8005$. Thus, the greedy policy performs $80.05\%$ worse than the optimal policy.

\subsection{On the monotonicity in the battery state}\label{sec:fading-battery}
Consider the model in Sec.~\ref{sec:battery} with $N_{0} = 1.55$, $W = 1.75$, and  an i.i.d.\@ fading channel
where $\mathcal H = \{1,2\}$, $g(\cdot) = \{0.75, 0.80\}$, and $P_H =
[0.3,0.7]$. The optimal policy for this model (obtained using policy
iteration) is shown in Fig.~\ref{fig:ex4a}--\ref{fig:ex4b}. Note that for $h
\in \{1,2\}$, the optimal policy is not monotone in the battery state.

In this case, there are $(629856) \times (30375019) \approx 10^{10}$ monotone policies. Therefore, a
brute force search is not possible. As before, we choose a heuristic policy
$f^\circ_s$ which is the battery-monotone policy that is closest to $f^*$. In
particular, $f^\circ_s$ differs from $f^*$ only at two points:
$f^\circ_s(5,3,1) = 1$ and $f^\circ_s(5,3,2) = 1$. Let $V^\circ_s$ denote the
corresponding value function. The worst case difference between the two value
functions is given by
\[
  \alpha_{s} = \max_{(n,s,h) \in \mathcal L \times \mathcal B \times \mathcal H}
    \frac{\bigl| V^*(n,s,h) - V^\circ_{s}(n,s,h)\bigr|}
    {\bigl| V^*(n,s,h) \bigr|}
  = 0.0560.
\]
Thus, the heuristically chosen battery-monotone policy performs $5.60\%$ worse than the optimal policy. Note that for this example, the best monotone policy is a greedy policy, hence the performance of the greedy policy is  same as that of the best monotone policy.

\section{Conclusion}
In this paper, we consider delay optimal strategies in cross layer design with
energy harvesting transmitter. We show that the value function is weakly
increasing in the queue state and weakly decreasing in the battery state. We
show via counterexamples that the optimal policy is not monotone in queue
length nor in the available energy in the battery. 

\subsection{Discussion about the counterexamples} 
One might ask why the optimal policy is not monotone in the above model. The
standard argument in MDPs to establish monotonicity of the optimal policies is
to show that the value-action function is submodular in the state and action.
The value-action function is given by 
\begin{equation*}
  {H}(n,s,u) = d(n-u) + 
  \beta\EXP\big[ V([n - u + A]_{L}, [s - p(u) + E]_{B}) \big]
\end{equation*}
A sufficient condition for the optimal policy to be weakly increasing in the
queue length is:
\begin{enumerate}
  \item[\textbf{(S1)}] for every $s \in \mathcal{B}$, $H(n,s,u)$ is submodular in
    $(n,u)$.
\end{enumerate}
Note that since $d(\cdot)$ is convex, $d(n-u)$ is submodular in $(n,u)$. Thus, a
sufficient condition for (S1) to hold is:
\begin{enumerate}
  \item[\textbf{(S2)}] for all $s\in \mathcal{B}$, $\EXP\big[ V([n - u + A]_{L}, [s -
    p(u) + E]_{B}) \big]$ is submodular in $(n,u)$.
\end{enumerate}
Since submodularity is preserved under addition, a sufficient condition for
(S2) to hold is:
\begin{enumerate}
  \item[\textbf{(S3)}] for all $s \in \mathcal{B}$, $V(n-u, s-p(u))$ is submodular in
    $(n,u)$. 
\end{enumerate}    

By a similar argument, it can be shown that a sufficient condition for the
optimal policy to be weakly increasing in battery state is: 
\begin{enumerate}
  \item[\textbf{(S4)}] for all $n \in \mathcal{L}$, $V(n-u, s-p(u))$ is submodular in $(s,u).$
\end{enumerate}

We have not been able to identify sufficient conditions under which (S3) or
(S4) hold. Note that if the data were backlogged, then we do not need to keep
track of the queue state; thus, the value function is just a function of the
battery state. In such a scenario, (S4) simplifies to $V(s-p(u))$ is
submodular in $(s,u)$. Since $p(\cdot)$ is convex, it can be shown that
convexity of $V(s)$ is sufficient to establish submodularity of $V(s-p(u))$.
This is the essence of the argument given in~\cite{mao2014joint,sinha2012optimal}.

Similarly, if the transmitter had a steady supply of energy, then we do not
need to keep track of the battery state; thus, the value function is just a
function of the queue state. In such a scenario, (S3) simplifies to $V(n-u)$
is submodular in $(n,u)$. It can be shown that convexity of the $V(n)$ is
sufficient to establish submodularity of $V(n-u)$. This is the essence of the
argument given in~\cite{berry2000power}.

In our model, data is not backlogged and energy is intermittent. As a result,
we have two queues---the data queue and the energy queue---which have coupled
dynamics. This coupling makes it difficult to identify conditions under which
$V(n-u,s-p(u))$ will be submodular in $(n,u)$ or $(s,u)$.

\subsection{Implication of the results}

In general, there are two benefits if one can establish that the optimal
policy is monotone. The first advantage is that monotone policies are easier
to implement. In particular, one needs a $(L+1) \times (B+1)$-dimensional
look-up table to implement a general transmission policy (similar to the
matrices shown in Figs.~\ref{fig:ex1} and~\ref{fig:ex2}). In contrast, one
only needs to store the thresholds boundaries of the decision regions (which can be stored in a sparse matrix) to implement a queue- or battery-monotone
policy. Our counterexamples show that such a simpler implementation will
result in a loss of optimality in energy-harvesting systems.

The second advantage is that if we know that the optimal policy is monotone,
we can search for them efficiently using monotone value iteration and monotone
policy iteration~\cite{puterman2014markov}. Our counterexamples show that
these more efficient algorithms cannot be used in energy-harvesting systems.

One might want to restrict to monotone policies for the sake of implementation
simplicity. However, if the system does not satisfy properties (S3) and (S4)
mentioned in the previous section, then dynamic programming cannot be used to
find the best monotone policy. Thus, one has to resort to a brute force
search, which suffers from the curse of dimensionality.

\appendix

\subsection{Monotonicity of Bellman operator}

\begin{lemma} \label{lem:monotone}
  Given $V \colon \mathcal{L} \times \mathcal{B} \to \reals$, define $H
  \colon \mathcal{L} \times \mathcal{B} \times \mathcal{U} \to \reals$: 
  \begin{equation*}
    {H}(n,s,u) = d(n-u) + 
    \beta\EXP\big[ V([n - u + A]_{L}, [s - p(u) + E]_{B}) \big].
  \end{equation*}
  If $V \in \mathcal{M}$, then for all $n \in \mathcal{L}$, $s
  \in \mathcal{B}$, and $u \in \mathcal{U}(n,s)$:
  \begin{enumerate}
    \item ${H}(n,s,u) \leq {H}([n+1]_L,s,u)$. 
    \item ${H}(n,s,n) \leq {H}([n+1]_L,s,[n+1]_L)$.
    \item ${H}(n,[s+1]_B,u) \leq {H}(n,s,u)$.
  \end{enumerate}
  Consequenly, $\mathscr{B}V \in \mathcal{M}$.
\end{lemma}

\begin{proof}
  The properties of $H$ follow from the monotonicity of $d(\cdot)$ and the
  fact that monotonicity is preserved under expectations. The details are
  omitted due to lack of space.

  To prove that $W = \mathscr{B}V \in \mathcal M$, we consider any $n \in \mathcal L$ and $s
  \in \mathcal B$ and let $f(n,s)$ denote a policy that achieves the minimum
  in the definition of $\mathscr{B}V$. There are two cases: 
  $f({n+1},s) \neq n+1$ and $f(n+1,s) = n+1$.
  \begin{enumerate}
    \item Suppose $u^* = f(n+1,s) \neq n+1$. Then, it must be the case that
      $u^* \in \mathcal U(n,s)$. Thus,
      \begin{align*}
        W(n+1,s) &= H(n+1, s, u^*) 
        \stackrel{(a)}\ge H(n, s, u^*) \\
        & \ge \min_{u \in \mathcal U(n,s)} H(n,s,u) 
        = W(n,s),
      \end{align*}
      where $(a)$ follows from Property~1.

    \item Suppose $u^* = f(n+1,s) = n+1$. Then, it must be the case that
      $p(n+1) \le s$ and, therefore, $p(n) \le s$. Hence $n \in \mathcal
      U(n,s)$. Thus,
      \begin{align*}
        W(n+1,s) &= H(n+1,s,n+1) 
        \stackrel{(b)}\ge H(n,s,n) \\
        & \ge \min_{u \in \mathcal U(n,s)} H(n,s,u) 
        = W(n,s),
      \end{align*}
      where $(b)$ follows from Property~2.
  \end{enumerate}
  As a result of both of these cases, we get that
  \begin{align}\label{pf:1}
    W(n,s) \leq W(n+1,s).
  \end{align}
\item Now let $u^{*} = f(n,s)$, recall that $\mathcal{U} (n,s) \subseteq \mathcal{U} (n,s+1)$ then $u^{*} \in \mathcal{U}(n,s+1)$ thus
  \begin{align}\label{pf:2}
    W(n,s) &= H(n, s, u^*)
    \stackrel{(c)}\ge H(n, s+1, u^*)\notag \\
    & \stackrel{(d)}\ge \min_{u \in \mathcal U(n,s+1)} H(n,s+1,u)
    = W(n,s+1),
  \end{align}
  Where $(c)$ follows from Property~3 and $(d)$ follows from the fact that $u^{*} \in \mathcal{U}(n,s+1)$.

  From \eqref{pf:1} and \eqref{pf:2} we infer $W \in \mathcal{M}$.
\end{proof}

\subsection{Proof of Proposition~\ref{thm:monotone}}

Arbitrarily initialize $V^{(0)} \in \mathcal M$ and
for $n \in \integers_{>0}$, recursively define 
\(
  V^{(n+1)} = \mathscr{B} V^{(n)}.
\)
Since $V^{(0)} \in \mathcal M$, Lemma~\ref{lem:monotone} implies that 
$V^{(n)} \in \mathcal{M}$, for all $n \in
\integers_{>0}$. Since monotonicity is preserved under the limit, we have that
$\lim_{n \rightarrow \infty} V_{0}^{(n)} \in \mathcal{M}$.
By~\cite{puterman2014markov},
\(
  \lim_{n \rightarrow \infty} V_{0}^{(n)} = V.
\)
Hence, $V \in \mathcal{M}$.




\end{document}